\documentclass[letterpaper, 10 pt, conference]{ieeeconf}  

\IEEEoverridecommandlockouts                              
\usepackage{graphics} 
\usepackage{graphicx}
\usepackage{booktabs} 
\usepackage{amsmath, amssymb, amsthm}
\usepackage{caption, subcaption}
\usepackage{breqn}
\usepackage{array, mathtools}
\usepackage{siunitx}
\usepackage{algorithm}
\usepackage[noend]{algpseudocode}
\usepackage{tabulary}
\usepackage[usenames, dvipsnames]{color}
\usepackage[colorinlistoftodos]{todonotes}
\usepackage{soul}
\usepackage{cite}

\usepackage{tikz}
\usepackage{pgfplots}
\usetikzlibrary{decorations.markings}
\usetikzlibrary{patterns}
\usetikzlibrary{arrows}
\usetikzlibrary{shapes}
\usetikzlibrary{fit}
 \usetikzlibrary{calc}
\usetikzlibrary{decorations.pathreplacing}
\usetikzlibrary{arrows,positioning} 
\usetikzlibrary{pgfplots.groupplots}

\newtheorem{proposition}{Proposition}

\newtheorem{theorem}{Theorem}

\theoremstyle{definition}
\newtheorem{definition}{Definition}
\newtheorem{remark}{Remark}
\newtheorem{example}{Example}
\newtheorem{problem}{Problem}

\newcommand{\R}{\mathbb{R}}
\newcommand{\N}{\mathbb{N}}
\newcommand{\A}{\mathcal{A}}

\title{\LARGE \bf Lyapunov Differential Equation Hierarchy and Polynomial Lyapunov Functions for Switched Linear Systems}

\author{Matthew Abate, Corbin Klett, Samuel Coogan, and Eric Feron
\thanks{This material is based upon work supported by the United States Government under Air Force Office of Scientific Research grant number FA9550-19-1-0015.  Any opinions, findings and conclusions or recommendations expressed in this material are those of the author(s) and do not necessarily reflect the views of the Government.}
\thanks{M. Abate is with the School of Mechanical Engineering and the School of Electrical and Computer Engineering, Georgia Institute of Technology, Atlanta, 30332, USA: {\tt\small Matt.Abate@GaTech.edu}.}
\thanks{C. Klett and E. Feron are with the School of Aerospace Engineering, Georgia Institute of Technology, Atlanta, 30332, USA: 
{\tt\small Corbin@GaTech.edu} and {\tt\small Feron@GaTech.edu}.}
\thanks{S. Coogan is with the School of Electrical and Computer Engineering and the School of Civil and Environmental Engineering, Georgia Institute of Technology, Atlanta, 30332, USA: {\tt\small Sam.Coogan@GaTech.edu}.}
}

\begin{document}

\maketitle
\thispagestyle{empty}
\pagestyle{empty}

\begin{abstract}
This work studies the problem of searching for homogeneous polynomial Lyapunov functions for stable switched linear systems.
Specifically, we show an equivalence between polynomial Lyapunov functions for systems of this class and quadratic Lyapunov functions for a related hierarchy of Lyapunov differential equations. 
This creates an intuitive procedure for checking the stability properties of switched linear systems, and a computationally competitive algorithm is presented for generating high-order homogeneous polynomial Lyapunov functions in this manner. Additionally, we provide a comparison between polynomial Lyapunov functions generated with our proposed approach and polynomial Lyapunov functions generated with a more traditional sum-of-squares based approach.
\end{abstract}

\section{Introduction}
Switched dynamical system models appear throughout the field of control theory, and the structure of such models has been widely explored and exploited in order to the analyze stability and performance of real-world systems \cite{liberzon2003switching, feron1996quadratic}. In turn, such results have inspired the use of switched systems as a modeling tool for many challenging analysis problems. For example, hybrid dynamical systems can be represented as switched systems, as can some stochastic systems \cite{4806347, Brockett1993}. Certain nonlinearities such as saturation and mechanical backlash can be modeled using switched linear systems \cite{liberzon2003switching, yak1, BEFB:94, Barmish1985}, as can random noise \cite{yoon2019}. Additionally, switched linear systems can be used as an over-approximating abstraction for more general nonlinearities \cite{yak1, yak2} and, for this reason, switched linear system models appear widely in robustness analysis literature \cite{BEFB:94, yak3}. 
Further, the consistent use of switched system models in safety-critical applications has facilitated the need for computationally efficient analysis tools.

Stability-type proofs for switched dynamical systems often require the construction of polynomial Lyapunov functions.
Such proofs guarantee system stability by associating a global energy field with the system state space and then showing that energy is decreasing for all initial conditions and all switched modes.  The simplest class of polynomial Lyapunov function is the class of \textit{quadratic} Lyapunov functions and, as such, the search for quadratic Lyapunov functions has computational advantages in comparison to other methods of stability analysis. Numerous works, including \cite{Barmish1985, rantzer1997} explore the guarantees attainable when solely searching for quadratic Lyapunov functions, however, recent progress in sum-of-squares based techniques have shown that higher-order polynomial Lyapunov functions can be calculated as well with more accurate stability guarantees \cite{parrilo2000, parrilo2003semidefinite}.  In general, sum-of-squares based techniques require little machinery to implement; these methods cast the search for a polynomial Lyapunov function as a convex feasibility problem, and many efficient solvers exist to solve such problems \cite{parrilo2000}. Additionally, for systems which are known to be stable, the computation of high-order polynomial Lyapunov functions has the ability to help characterize invariant regions of the state space with complex geometries; this is not possible when computing quadratic Lyapunov functions.

This work provides an algorithm for constructing homogeneous polynomial Lyapunov functions for switched linear systems. The aforementioned algorithm searches for polynomial Lyapunov functions through a convex feasibility problem, however, the structure of our algorithm differs significantly from traditional sum-of-squares formulations.  Specifically, we encode the search for polynomial Lyapunov functions as a search for quadratic Lyapunov functions for a related hierarchy of Lyapunov differential equations. This creates an intuitive procedure for checking the stability properties of switched linear systems and enables new applications as well \cite{yoon2019}.
Moreover, we show that every homogeneous sum-of-squares polynomial Lyapunov function for a given initial system can be transformed to a quadratic polynomial Lyapunov function for a system in the related hierarchy; this procedure can also be conducted in the reverse order, allowing one to generate sum-of-squares polynomial Lyapunov functions for an initial system through the identification of a quadratic polynomial Lyapunov function for a related system.

This paper is organized in the following way.  We review common analysis tools for assessing the stability of switched linear systems in Section II. Specifically, we introduce a time-varying Lyapunov differential equation, which we define in reference to an initial switched linear system.  Using the time-varying Lyapunov differential equation as an initial case, we then form a hierarchy of Lyapunov differential equations in Section III. Quadratic Lyapunov functions for differential equations in this hierarchy are shown to correspond to homogeneous polynomial Lyapunov functions for the initial switched system later in the same section.  Section IV explores the relation between quadratic Lyapunov functions for the aforementioned hierarchy of Lyapunov differential equations and homogeneous sum-of-squares polynomial Lyapunov functions for the initial switched linear system.
Finally, we provide an algorithm, formulated as a convex optimization problem, for computing high-order homogeneous polynomial Lyapunov functions for switched linear systems in Section V; this algorithm is presented in conjunction with a numerical example.


\section{Stability and Switched Linear Systems}
\subsection{Preliminaries}
Consider the linear time-variant system
\begin{equation}
  \dot{x}=A(t)x,
  \label{eqn1}
\end{equation}
\noindent{}where $x(t)\in\mathbb{R}^n$ denotes the system state, and $A(t) \in\R^{n\times n}$ evolves nondeterministically inside a finite set of switched linear modes $A(t) \in \left\{A_1,\, \cdots,\, A_N\right\}$. We assume that each of the switched modes $\dot{x} = A_i x$, with $i\in\{1,\cdots,N\}$, converges asymptotically to the origin for all initial conditions $x(0) \in \R^n$.

Importantly, the asymptotic stability of each mode does not, by itself, imply the asymptotic stability of the system \eqref{eqn1} under arbitrary switching; see \cite{liberzon2003switching} Chapter 2 for further details.
As such, more complex techniques are required to analyze the stability of \eqref{eqn1}. 

In this work, we consider a traditional approach for stability analysis for switched linear systems, involving the search for a common polynomial Lyapunov function that stabilizes each switched mode (Definition \ref{def:clf}).

\begin{definition}\label{def:clf}
A common Lyapunov function for the system \eqref{eqn1} is a mapping $V: \mathbb{R}^n \rightarrow \mathbb{R}$ such that 
\begin{equation}
    \label{LyapCond}
    \begin{array}{rcl}
    V(x) &>& 0 \\
    \dot{V}(x)= \langle \nabla V, A_i x \rangle &<& 0 \\
    \forall x & \neq & 0  \\
    \forall i\,\: & \in  &\{1,\,\cdots,\, N\}.
    \end{array}
\end{equation}
\end{definition}

It is well known that the system \eqref{eqn1} is stable if and only if there exists a Lyapunov function $V(x)$ which satisfies \eqref{LyapCond}.
 Moreover, the authors of \cite{CPLF} show that \eqref{eqn1} is stable if and only if there exists a common \textit{homogeneous} polynomial Lyapunov function which proves stability of each mode.  We capture this assertion in Remark \ref{remark1}.

\begin{remark}\label{remark1}
\cite[Theorem 4.5]{6161493} If the switched linear system \eqref{eqn1} is asymptotically stable under arbitrary switching, then there exists a polynomial Lyapunov function $V(x)$, satisfying \eqref{LyapCond}, which is homogeneous in the entries of $x$.
\end{remark}


\subsection{Quadratic Lyapunov Functions for Switched Systems}\label{2b}
Reconsider the system \eqref{eqn1}. In the special instance that there exists a $V(x)$, satisfying \eqref{LyapCond}, which is quadratic in the entries of $x$, we say that the system \eqref{eqn1} is \textit{quadratically stable} \cite{Barmish1985}.
Such a Lyapunov function will take the form
\begin{equation*}
    V(x) = x^TPx 
\end{equation*}
where $P \in \mathbb{R}^{n\times n}$ is a symmetric positive definite matrix and
\begin{equation}
    A_i^T P + P A_i \,<\, 0
    \label{eq:eqn3}
\end{equation}
for all $i \,\in\, \{1,\,\cdots,\,N\}$.
Alternatively, one can show that the system \eqref{eqn1} is quadratically stable by showing that there exists a symmetric positive definite $Q\in \mathbb{R}^{n\times n}$ that satisfies
\begin{equation}
A_i Q + Q A_i^T \,<\, 0
\label{eqn4}
\end{equation}
for all $i \,\in\, \{1,\,\cdots,\,N\}$ \cite{Barmish1985}; in this case,
\begin{equation*}
    V(x) = x^T Q^{-1} x
\end{equation*}
is a quadratic Lyapunov function for the system \eqref{eqn1}.

Quadratic polynomial Lyapunov functions are the simplest substantiation of homogeneous polynomial Lyapunov functions, and thus, the search for a quadratic Lyapunov function for \eqref{eqn1} has computational advantages in comparison to other strategies for stability analysis; the search can be reduced to solving a convex feasibility problem involving linear matrix inequalities, and many efficient solvers exist to solve such problems \cite{BEFB:94, SeDuMi}. 
Recent progress in polynomial optimization systems via sum-of-squares relaxations, however, has shown that more general polynomial Lyapunov functions could be computed as well with added benefits, such as improved system stability margins.

Importantly, if the system \eqref{eqn1} is linear time-invariant, i.e. $N = 1$, then \eqref{eqn1} is asymptotically stable if and only if there exists a $P, Q \in \mathbb{R}^{n\times n}$ satisfying \eqref{eq:eqn3} and \eqref{eqn4}, respectively. This is not true, however, in the case of multiple switched modes; stable switched linear systems exist for which there is no quadratic Lyapunov function certifying the stability of each mode  \cite[Section 3]{rantzer1997}.  For this reason, we must resort to more complex tools to prove stability in the general setting of \eqref{eqn1}.


\subsection{The Lyapunov Differential Equation}
We next present the time-variant switched Lyapunov differential equation:
\begin{equation}
\dot{X} = A(t)X + X A(t)^T,
\label{eqn5}
\end{equation}
where $X(t)\in \mathbb{R}^{n\times n}$ and $A(t)$ retains its definition from \eqref{eqn1}.
In this work, we primarily use the Lyapunov differential equation \eqref{eqn5} as a stability analysis tool for the initial switched system \eqref{eqn1}. As is shown in the following proposition, \eqref{eqn5} is stable if and only if \eqref{eqn1} is stable; moreover, stability guarantees on the Lyapunov differential equation propagate down to stability guarantees on the initial system.

\begin{proposition}\label{prop1}
The switched Lyapunov differential equation \eqref{eqn5} is stable if and only if the system \eqref{eqn1} is stable.
\end{proposition}
\begin{proof}
{\Large $\bigcirc$}\hspace{-4.5mm}$\Rightarrow$  Assume the system \eqref{eqn5} is stable, and let $X=xx^T$. Then 
\begin{equation*}
\begin{split}
    \dot{X} 
    & = \dot{x}x^T + x \dot{x}^T \\
    & = A(t)xx^T + xx^TA(t)^T \\
    & = A(t)X+X A(t)^T.
\end{split}
\end{equation*}
Therefore $X=xx^T$ converges to zero, which implies $x$ converges to zero along trajectories of \eqref{eqn1}.

\noindent{}{\Large $\bigcirc$}\hspace{-4.5mm}$\Leftarrow$ Assume the system \eqref{eqn1} is stable, and define $X(t)\in \mathbb{R}^{n\times n}$ with initial condition $X(0) = X_0$. Any matrix can be written as the sum of diads; therefore, there exist $p_{1,0},\, \cdots,\, p_{N,0}, q_{1,0},\, \cdots,\, q_{N,0} \in\mathbb{R}^n$ such that
    \begin{equation*}
        X_0 = \sum_{j=1}^N p_{j,0}q_{j,0}^T.
    \end{equation*}
Next, consider the $2N$ trajectories that satisfy 
\begin{equation}
\frac{d}{dt}p_j(t) = A(t) p_j(t), \quad p_j(0) = p_{j,0},
\label{diff_eq_1}
\end{equation}
\begin{equation}
\frac{d}{dt}q_j(t) = A(t) q_j(t), \quad q_j(0) = q_{j,0}. 
\label{diff_eq_2}
\end{equation}
where $j \in \{1,\, \cdots,\, N\}$, and note that if \eqref{eqn1} is stable then \eqref{diff_eq_1} and \eqref{diff_eq_2} converge to zero.

Taking $X=\sum_{j=1}^N p_j(t)q_j(t)^T$ then yields
\begin{equation*}
    \begin{split}
        \dot{X}(t)
        & = \sum_{j=1}^N\Big{(}\dot{p}_j q_j^T + p_j\dot{q}_j^T\Big{)}\\
        & = A(t)\bigg{(}\sum_{j=1}^N p_j q_j^T\bigg{)} + \bigg{(}\sum_{j=1}^N p_j q_j^T\bigg{)}A(t)^T\\
       & = A(t) X(t) + X(t) A(t)^T.
    \end{split}
\end{equation*}
Therefore, $X=\sum_{j=1}^N p_j(t)q_j(t)^T$ is a (unique) solution to the differential equation \eqref{eqn5} with initial condition $X_0$, and $p_j(t)$ and $q_j(t)$ are stable for all $j\in \{1,\,\cdots,\,N\}$. Therefore, $X(t)$ also converges along trajectories of \eqref{eqn5}.
\end{proof}


\section{Establishing a Hierarchy of Lyapunov Differential Equations}

In this section we build on \eqref{eqn5} to create a hierarchy of Lyapunov differential equations for the system \eqref{eqn1}. As was the case in Proposition \ref{prop1}, each system in the hierarchy is shown to have equivalent stability properties.


\subsection{Notation}
Let $A \otimes B \in \R^{np\times mq}$ denote the Kronecker product of $A \in \R^{n\times m}$ and $B \in \R^{p \times q}$. Let $\otimes^k x \in \R^{n^k}$ denote the $k^{\text{th}}$ Kronecker power of $x \in \R^n$, which is defined recursively by
\begin{equation*}
    \begin{array}{rllc}
         \otimes^1 x & = x & \in \R^n, & \\
         \otimes^{k} x & = x \otimes (\otimes^{k-1} x) & \in \R^{n^k}, & k \geq 2.
    \end{array}
\end{equation*}
Let $W^+ \in \R^{m\times n}$ denote the Moore-Penrose inverse of $W \in \R^{n\times m}$, and let $I_n \in \R^{n \times n}$ denote the $n \times n$ identity matrix.


\subsection{Identifying Meta-Lyapunov Functions}\label{threeA}

We first rewrite \eqref{eqn5} as
\begin{equation}
\dot{\vec{X}} = {\cal A}(t) \vec{X}
\label{eqn8}
\end{equation}
by taking $\vec{X}$ to be the vectorization of $X$, i.e. $\vec{X} = \mbox{vec}(X) \in \mathbb{R}^{n^2}$. In this case, $\A(t) \in \R^{n^2 \times n^2}$ evolves nondeterministically in the set $\A(t) \in  \{\A_1,\cdots, \A_N\}$, where $\A_i$ is defined by 
\begin{equation*}
    \A_i := I_n \otimes A_i + A_i \otimes I_n
\end{equation*}
for $i \in \{1,\,\cdots,\, N\}$.

For convenience, we refer to \eqref{eqn8}, which is also linear time-variant, as the {\em meta-system} relative to system~\eqref{eqn1}.
Applying concepts of quadratic stability to meta-systems, the system \eqref{eqn8} is stable if there exists a positive definite $P\in \mathbb{R}^{n^2 \times n^2}$ such that 
\begin{equation} \label{eqn9}
    \mathcal{A}_i^T P + P \mathcal{A}_i < 0 ,
\end{equation}
for all $i \in \{1,\, \cdots,\, N\}$.  These constraints correspond to the existence of a Lyapunov function $\mathcal{V}(\vec{X}) = \vec{X}^T P\vec{X}$ for \eqref{eqn8}, which is quadratic in the entries of $\vec{X}$.  In what follows, we refer to $\mathcal{V}(\vec{X})$ as a \textit{meta-Lyapunov function} for the system \eqref{eqn1}, and we formalize the search for such a meta-Lyapunov function as the main inquiry of the section.
\begin{problem}\label{prob1}
Given a system \eqref{eqn1}, which is known to be stable, find a positive definite matrix $P \in \mathbb{R}^{n^2 \times n^2}$ that satisfies \eqref{eqn9}.
\end{problem}

As was shown in Proposition \ref{prop1}, if the system \eqref{eqn1} is stable, then there must be a Lyapunov function $\mathcal{V}(\vec{X})$ that certifies the stability of the meta-system \eqref{eqn8}; this Lyapunov function, however, need not be quadratic.  In what follows, we show that in the special instance that \eqref{eqn1} is quadratically stable, there must exist a quadratic Lyapunov function certifying the stability of the meta-system \eqref{eqn8}, and moreover, there must be a $P \in \R^{n^2 \times n^2}$ that solves Problem \ref{prob1}.  We capture this assertion in the following theorem.

\begin{theorem}\label{trm1}
If the system \eqref{eqn1} is quadratically stable, then the system \eqref{eqn8} is also quadratically stable. 
\end{theorem}
\begin{proof}
Assume there exists of a quadratic Lyapunov function $V(x) = x^TQx$ for \eqref{eqn1}, and pick $P = Q\otimes Q$. Indeed $P$ is positive definite in the instance $Q$ is positive definite.  Moreover,
\begin{equation*}
    \mathcal{V}(\vec{X}) \;:=\;  \vec{X}^T P\vec{X} \;>\; 0
\end{equation*}
for all nonzero $\vec{X} \in \mathbb{R}^{n^2}$.

We next show that $\mathcal{V}$ decreases along the trajectories of \eqref{eqn8}. From \eqref{eqn8} we have
\begin{equation*}
    \dot{\mathcal{V}}(\vec{X}) \;=\;  \vec{X}^T\Big{(}\mathcal{A}(t)^TP + P\mathcal{A}(t)\Big{)}\vec{X}.
\end{equation*}
Further, we calculate
\begin{equation*}
    \begin{split}
       \mathcal{A}(t)^T P 
        & = (I_n \otimes A(t)^T + A(t)^T \otimes I_n) Q\otimes Q \\
        & = Q \otimes (A(t)^TQ) + (A(t)^TQ) \otimes Q
    \end{split}
\end{equation*}
and
\begin{equation*}
    \begin{split}
        P \mathcal{A}(t)
        &= Q\otimes Q(I_n \otimes A(t) + A(t) \otimes I_n) \\
        &= Q \otimes (QA(t)) + (QA(t)) \otimes Q.
    \end{split}
\end{equation*}
Grouping terms then yields
\begin{multline*}
    \mathcal{A}(t)^T P + P \mathcal{A}(t)= Q \otimes (A(t)^TQ + QA(t)) + \cdots\\
    \qquad\qquad+ (A(t)^TQ + QA(t)) \otimes Q.
\end{multline*}

We now check that $\dot{\mathcal{V}}(\vec{X})$ is negative for all nonzero $\vec{X}\in \R^{n^2}$.
To that end, note that
\begin{multline*}
        \Vec{X}^T\Big{(}Q \otimes (A(t)^TQ + QA(t))\Big{)}\Vec{X} =\\
        \begin{array}{l}
            \cdots=  \Vec{X}^T \mbox{vec}\Big{(}(A(t)^TQ+QA(t))XQ\Big{)}\\
            \cdots=  \mbox{trace}\Big{(} X^T(A(t)^TQ+QA(t))XQ\Big{)} \\
            \cdots = \mbox{trace}\Big{(} Q^{1/2}X^T(A(t)^TQ+QA(t))XQ^{1/2}\Big{)},
        \end{array}
\end{multline*}
and
\begin{multline*}
        \Vec{X}^T\Big{(}(A(t)^TQ + QA(t)) \otimes Q\Big{)}\Vec{X} =\\
        \begin{array}{l}
             \cdots= \Vec{X}^T \mbox{vec}\Big{(}QX(A(t)^TQ+QA(t))\Big{)}\\
             \cdots= \mbox{trace}\Big{(} X^TQX(A(t)^TQ+QA(t))\Big{)} \\
             \cdots = \mbox{trace}\Big{(} Q^{1/2}X^T(A(t)^TQ+QA(t))XQ^{1/2}\Big{)}
        \end{array}
\end{multline*}
Since $A(t)^TQ+QA(t)$ is negative semidefinite, so is $Q^{1/2}X^T(A^TQ+QA)XQ^{1/2}$, and its trace is negative. 
Thus $\dot{\mathcal{V}}(\vec{X})$ is negative for all nonzero $\vec{X}\in \R^{n^2}$, and moreover, $\mathcal{V}(\vec{X}) = \vec{X}^T (Q\otimes Q)\vec{X}$ is a quadratic Lyapunov function certifying the stability of the meta system \eqref{eqn8}.
Additionally, this result confirms that $P = Q\otimes Q$ solves problem \ref{eqn1}.

\end{proof}

It is of course possible to repeat the process again and certify stability at a deeper level; for instance, one may form the Lyapunov differential equation corresponding to \eqref{eqn8},
\begin{equation}\label{metameta}
    \frac{d}{dt} \xi = (I\otimes \mathcal{A}(t)+\mathcal{A}(t)\otimes I) \xi,
\end{equation}
$\xi \in \R^{n^4}$ and then show that
\begin{equation*}
    V(\xi) = \xi^T (Q\otimes Q\otimes Q\otimes Q)\xi
\end{equation*}
is a quadratic Lyapunov function for the new meta-system \eqref{metameta}.
Pursuing the process further, it is possible to construct a ``hierarchy" of Lyapunov differential equations whose state space dimensions are $n^{2^c}$, where $c$ is an integer greater than or equal to $1$.
In the following section, we complete this hierarchy to include Lyapunov differential equations whose state space dimensions grow as $n^{2 c}$.


\subsection{A Linear Hierarchy of Polynomial Lyapunov Functions}
We next develop a hierarchy of dynamical systems whose state space dimensions grow as integer exponents of $n$, the dimension of the state space of \eqref{eqn1}.  This hierarchy complements the hierarchy of systems discussed above.


\begin{theorem}\label{thrm:two}
System \eqref{eqn1} is stable if there exists $c\in \mathbb{N}_{\geq 1}$ and $P_c \in \mathbb{R}^{n^{c} \times n^{c}}$ positive definite such that
\begin{equation}
    \mathcal{A}_{c,i}^T P_c + P_c\mathcal{A}_{c,i} < 0
    \label{eqn11}
\end{equation}
for all $i \in \{1,\,\cdots,\,N\}$, where 
\begin{equation}\label{eqn12}
    \mathcal{A}_{c,i} := \sum_{ j = 0}^{c-1} I_{n^j} \otimes A_i \otimes I_{n^{c-1-j}}.
\end{equation}
\end{theorem}
\begin{proof}
Taking $\vec{X} = \otimes^{c}x(t) \in \mathbb{R}^{n^c}$, we find
\begin{equation}\label{eqn13}
        \dot{\vec{X}} = \mathcal{A}_c(t) \vec{X}
\end{equation}
where $\mathcal{A}_c$ is given by \eqref{eqn11}, and the stability of system \eqref{eqn13} implies that of System \eqref{eqn1}.  Therefore System \eqref{eqn1} is stable if there exists a positive definite $P_c \in \mathbb{R}^{n^{c} \times n^{c}}$ such that \eqref{eqn11} holds.
\end{proof}

Theorem \ref{thrm:two} shows that the existence of a $P_c \in \mathbb{R}^{n^c\times n^c}$ satisfying \eqref{eqn11} for some integer $c \geq 1$ certifies the stability of \eqref{eqn1}; such a $P_c$ identifies 
\begin{equation}\label{eqn14}
    V_c(x) = \big{(}\otimes^{c}x(t)^T\big{)} P_c \big{(}\otimes^{c}x(t)\big{)}
\end{equation}
as a polynomial Lyapunov function for \eqref{eqn1}, which is homogeneous in the entries of $x$ and of order $2c$.
Importantly, the degree of $V_c(x)$ grows linearly with $c$.

\subsection{Reducing the Dimensionality of the Meta-System}
The benefits of searching for meta-Lyapunov functions for \eqref{eqn1} using the methods presented thus far are namely structural; \eqref{eqn11}-\eqref{eqn12} provide an intuitive procedure for generating high-order homogeneous polynomial Lyapunov function for \eqref{eqn1} and moreover, this procedure does not require any heavy machinery to implement.
In contrast, there are few computational advantages to this approach, at present.
This is due in part to internal redundancy built into the Lyapunov constraints given by \eqref{eqn11}.  We demonstrate this assertion through the following example.

\begin{example}\label{example2}
Consider, for example, the system \eqref{eqn1} evolving in $\R^2$.  In this case, $x = [x_1,\, x_2]^T \in \R^2$, and $\vec{X} := x \otimes x \in \R^4$ is given by
\begin{equation}\label{exsys}
    \vec{X} = \begin{bmatrix}x_1^2 & x_1 x_2& x_1 x_2 & x_2^2\end{bmatrix}^T.
\end{equation}
When beginning at an initial condition $\vec{X}(0) = x(0) \otimes x(0)$ and evolving along trajectories of the meta-system
\begin{equation*}
    \dot{\vec{X}} = (I_2\otimes A(t) + A(t)\otimes I_2) \vec{X},
\end{equation*}
we find that the second and third entries of $\vec{X}$ remain equal to one another, regardless of the switching policy. This is due to the construction of $(I_2\otimes A(t) + A(t)\otimes I_2)$. 

The methods presented thus far address the problem of searching for a meta-Lyapunov function $\mathcal{V}(\vec{X}) = \vec{X}^T P \vec{X}$ for the system \eqref{exsys}; the specific choice of $P\in \R^{4\times 4}$ will then correspond to a homogeneous polynomial Lyapunov function $V_2(x) = (\otimes^2 x)^T P (\otimes^2 x)$ for the system \eqref{eqn1}.
Here, it is apparent that the constraints on $P$, given by \eqref{eqn11}, contain internal redundancy; note, for instance, that one must compute the $10$ unique entries of $P\in \R^{4\times 4}$ in order to find $\mathcal{V}(\vec{X})$, whereas, the resulting Lyapunov function $V_2(x)$ will only be defined by $5$ unique monomials.

Now, consider a vector containing the second-order monomials of $x$, this time with no redundancy.  Specifically, consider $y(x) = [x_1^2,\,x_1x_2,\, x_2^2] \in \R^3$, and note that $\vec{X} = W y(x)$ where
\begin{equation*}
    W = 
    \begin{bmatrix}
    1 &0 & 0 \\ 0 & 1& 0 \\ 0& 1& 0 \\ 0& 0& 1
    \end{bmatrix}.
\end{equation*}
Using \eqref{eqn13} as a basis, the dynamics of $y(x)$ can be captured in closed form:
\begin{equation*}
    \dot{y} = W^+ \A_c(t) W y.
\end{equation*}
Therefore, one can now formulate the search for a fourth-order homogeneous polynomial Lyapunov function for \eqref{eqn1}, as the search for a quadratic Lyapunov function $\overline{V}(y) = y^T \overline{P}y$ that certifies the stability of $y$.  In this case, the resulting Lyapunov function will have the same number of terms, i.e. 5 distinct monomials, however this search will only require the identification of the $6$ unique entries of $\overline{P} \in \R^{3 \times 3}$.
\qed
\end{example}

As shown in the previous example, the constraints given by \eqref{eqn11} are redundant; that is, a quadratic Lyapunov function that certifies the stability of $\vec{X}$, as in \eqref{eqn13}, will individually certify the stability of each of the meta-system's states, whereas, a reduced order meta-Lyapunov function that stabilizes a subset of meta-system's states may be sufficient.

For this reason, we present a new formulation of the constraints \eqref{eqn11}-\eqref{eqn12} that contains no redundancy.  We begin with the following definition.

\begin{definition}[$\A_c$-Invariant Subspaces]
A subspace $S \subset \R^n$ is said to be \textit{$\A_c$-invariant} for \eqref{eqn13} if for every vector $v \in S$ and every matrix $\A_{c, i}$ with $i \in \{1,\,\cdots,\, N\}$ we have $\mathcal{A}_{c,\,i} v \in S$.
\end{definition}

Note that $\A_{c}(t)$ as in \eqref{eqn13} will have an inherent invariant subspace, resulting from its construction.  We therefore remove this redundancy by analysing a reduced order meta-system, whose states correspond to unique monomials of the initial switched system \eqref{eqn1}. While the initial meta-Lyapunov conditions \eqref{eqn11} are defined by $n^{2c}$ constraints per switched mode, our new formulation only requires $M(n,\, c)^2$ such constraints, where $M(n,\, c)$ denotes the number of monomials of order $c \in \N_{\geq 1}$ in the entries of $x \in \R^n$ and is given by 
\begin{equation*}
    M(n,\, c) = {c + n - 1 \choose n - 1}.
\end{equation*}
This result is encapsulated in the following theorem.

\begin{theorem}\label{thrm3}
Let $y_c(x) \in \R^{M(n,\, c)}$ denote a vector containing the monomials of $x$ of order $c$, which we define in conjunction with a matrix $W_c \in \R^{n \times M(n,\, c)}$:
\begin{equation}\label{eqn16}
    \otimes^c x = W_c y_c(x).
\end{equation}
Additionally, let $\overline{P}_c \in R^{M(n,\,c) \times M(n,\,c)}$ be symmetric positive definite. If
\begin{equation}\label{eqn17}
    B^T_{c,\, i} \overline{P}_c + \overline{P}_c B_{c, i} < 0
\end{equation}
for all $i \in \{1,\, \cdots,\ N\}$ where
\begin{equation}\label{eqn18}
    B_{c,\, i} := W_c^{+} \A_{c,\, i} W_c,
\end{equation}
then $V_c(x) = y_c(x)^T \,\overline{P}\, y_c(x)$ is a homogeneous polynomial Lyapunov function for \eqref{eqn1} of order $2c$.
\end{theorem}

The proof of this result comes from the fact that $A_{c}(t)$ has an inherent invariant subspace, resulting from its construction.  As trajectories of 
\begin{equation*}
    \frac{d}{dt}(\otimes^c x) = \mathcal{A}_c(t) (\otimes^c x)
\end{equation*}
are known to begin in this subspace, we can encode the search for meta-Lyapunov functions for the system \eqref{eqn1} as a search for quadratic Lyapunov functions for the reduced order system 
\begin{equation}\label{eqn:ydyn}
    \dot{y}_c = B_c(t)y_c
\end{equation}
where $B_c(t) \in \{B_{c, 1},\, \cdots,\,B_{c, N}\}$ and $y(x)$ is given by \eqref{eqn16}.

Importantly, Theorem \ref{thrm3} allows the system designer to select $y_c(x)$ with whatever ordering properties they like; that is, we do not assume an order to the monomials that are stored in $y_c(x)$.  However, each ordering will induce a unique $W_c$, and thus the resulting Lyapunov conditions will always be the same, regardless of the chosen ordering. Moreover, the constraints given by \eqref{eqn17} are equivalent to the constraints given by \eqref{eqn11}, now with reduced dimensionality.

In the specific case where $n=2$, there is an intuitive ordering to the monomials of $x$; under this assumed ordering, the matrix $W_c$, given by \eqref{eqn16}, can be captured in closed form.

\begin{proposition}\label{prop2}
Consider the system \eqref{eqn1} and let $n=2$.  In this case we have
\begin{equation*}
    M(n, c) = c+1.
\end{equation*}
Additionally, for a positive integer $k \in \N_{\geq 0}$, let $0_{k} \in \R^{k}$ denote a vector populated with zeros.

If $y_c(x) \in \R^{c+1}$ conforms to the ordering
\begin{equation*}
    y_c(x) = \begin{bmatrix}x_1^{c} & x_1^{c-1}x_2 & \cdots & x_2^{c} \end{bmatrix}^T,
\end{equation*}
then we have $\otimes^c x = W_c\, y_c(x)$, where for an integer $k \in \N_{\geq 1}$ we define $W_k$ recursively by
\begin{equation} \label{eqn19}
\begin{array}{l}
     W_1 = I_2 \\
     W_{k} =
     \left[
\begin{array}{cc}
W_{k-1} & 0_{2^{k-1}} \\
\hline
0_{2^{k-1}} & W_{k-1} 
\end{array}
\right] \quad k\geq 2.
\end{array}
\end{equation}
\end{proposition}

In the case when $n > 2$, it is generally difficult to order the $c^{\text{th}}$ order monomials of $x$ in an intuitive way.  For this reason, we do not expand Proposition \ref{prop2} to account for the case where $n > 2$, nor do we suggest a canonical ordering for the entries of $y_c(x)$.  However, $W_c$ can always be solved for using \eqref{eqn16} once $y_c(x)$ has been chosen.


\section{Relation to Homogeneous Polynomial Lyapunov Functions}
Traditionally, the search for a polynomial Lyapunov functions systems of the form \eqref{eqn1} is encoded as the search for a sum-of-squares polynomial $V(x)$, satisfying \eqref{LyapCond}. 
\begin{definition}
A polynomial $p(x)$ is a \textit{sum-of-squares} in $x$ if there exist polynomials $g_1, \cdots, g_r$ such that 
\begin{equation*}
    p(x) = \sum_{i=1}^r g_i(x)^2.
\end{equation*}
\end{definition}
The search for a sum-of-squares polynomial $V(x)$, satisfying \eqref{LyapCond}, is known to be a convex optimization problem, computable by solving a semidefinite program \cite{parrilo2003semidefinite}.
Many efficient solvers exist to handle such problems \cite{BEFB:94, SeDuMi}.

We next show that the existence of quadratic Lyapunov functions for the hierarchy of dynamical systems \eqref{eqn13} guarantees the existence of a homogeneous sum-of-squares polynomial Lyapunov functions for \eqref{eqn1}, and vice versa. In this sense, all homogeneous sum-of-squares polynomial Lyapunov functions can be thought of as quadratic Lyapunov functions for a related hierarchy of differential equations.  Moreover, one can encode the search for high-order sum-of-squares polynomial Lyapunov functions, which certify the stability of \eqref{eqn1}, as a search for quadratic Lyapunov functions for the related system \eqref{eqn13}.  Calculating sum-of-squares polynomial Lyapunov functions in this way can be used to reduce the amount of machinery required to certify the stability of general switched linear systems of the form \eqref{eqn1}.

\begin{theorem}
There exists a $\overline{P}_c \in \mathbb{R}^{M(n,\,c) \times M(n,\,c)}$ satisfying \eqref{eqn17} for some positive integer $c \in \mathbb{N}_{\geq 1}$, if and only if there exists a homogeneous sum-of-squares polynomial Lyapunov function $V_{c}(x)$ of degree $2c$ for the system \eqref{eqn1}. 
\end{theorem}
\begin{proof}
A sum-of-squares polynomial that is homogeneous in the entries $x$ and of order $2c$ will take the form $p(x) = y_c(x)^T Z y_c(x)$, where $Z \in \R^{M(n, c) \times M(n, c)}$ is symmetric, and $y_c(x)$ and $M(n, c)$ retain their definitions from Theorem \ref{thrm3}.  From Theorem \ref{thrm3}, we have that if $\overline{P}_c \in \mathbb{R}^{M(n,\,c) \times M(n,\,c)}$ satisfies \eqref{eqn11} for some positive integer $c \in \mathbb{N}_{\geq 1}$, then we have that $V_c(x) = y_c(x)^T \overline{P} y_c(x)$ is a homogeneous polynomial Lyapunov function for \eqref{eqn20} and, moreover, $V_c(x)$ is a sum-of-squares.
To prove the converse, we note that if $p(x) = y_c(x)^T Z y_c(x)$ is a homogeneous sum-of-squares polynomial Lyapunov function for \eqref{eqn20} then $Z > 0$ and $\dot{p}(x) < 0$ for all $x \in \R^n$.  From the dynamics of $y_c(x)$, given as \eqref{eqn:ydyn}, we have $B^T_{c,\, i} Z + Z B_{c, i} < 0$ for all $i \in \{1, \cdots, N\}$.  Therefore $\overline{P}_c = Z$ solves \eqref{eqn17}.
\end{proof}


\section{Numerical Example}
In this section, we provide an example case and prove the stability of a switched linear system using a meta-Lyapunov function based approach.  An algorithm is provided for generating homogeneous polynomial Lyapunov functions for switched systems, which follows the procedure detailed in Theorem \ref{thrm3}; this algorithm is specifically written for implementation with CVX, a convex optimization toolbox made for use with MATLAB \cite{cvx}.  We also provide a comparison to a similar search for homogeneous polynomial Lyapunov functions that was implemented using SOSTOOLS, a sum-of-squares optimization toolbox made for use with MATLAB \cite{sostools}.  Experimental results are provided from MATLAB 2019b, which was run on a 2017 Macbook Pro laptop.

\subsection{Problem Formulation}
We consider the linear time-variant system
\begin{equation}\label{eqn20}
    \dot{x} = A(t) x \qquad A(t) \in \{A_1,\, A_2\} 
\end{equation}
\begin{equation*}
    A_1 = \begin{bmatrix}-.5 &.5 \\ -.5 & -.5\end{bmatrix} \qquad
    A_2 = \begin{bmatrix}-2.5 &2.5 \\ -2.5 & 1.5\end{bmatrix}.
\end{equation*}
In the following, we go about showing that \eqref{eqn20} is stable.  This is done, at first, through the computation of a quadratic Lyapunov function $V_1(x) = x^TPx$, which satisfies \eqref{LyapCond}, and then through the computation of higher-order homogeneous polynomial Lyapunov functions using the procedure detailed in Theorem \ref{thrm3}.

Importantly, if the system \eqref{eqn20} begins at an initial position $x_0 = x(0)$, and there exists a Lyapunov function $V_c(x)$ that certifies the stability of \eqref{eqn20}, then the infinite-time system trajectory is constrained to stay inside
\begin{equation}\label{idk}
    x(t) \in \{x \in \R^n \,\vert\, V_c(x) \leq V_c(x_0)\}
\end{equation}
for all $t \geq 0$.  For this reason, we select $V_c(x)$ as the minimizers of a suitable objective function, as to shrink the resulting invariant region derived through \eqref{idk}.
In what follows, we additionally show that computing higher-order meta-Lyapunov functions allows one to characterise tighter invariant sets by \eqref{idk}, even when the same objective function is used in each computation.

\subsection{Identifying Meta-Lyapunov Functions}
We search for meta-Lyapunov functions for \eqref{eqn20} using a semidefinite program.  
Specifically, when searching for a homogeneous Lyapunov function of order $2c$, we first calculate $B_{c, 1}$ and $B_{c, 2}$ using equations \eqref{eqn12}, \eqref{eqn18} and \eqref{eqn19}, and then we search for a symmetric positive-definite matrix $\overline{P}_c \in \R^{(c+1) \times (c+1)}$ that satisfies \eqref{eqn17}.  Such a matrix identifies $V_{c}(x) = y_c(x)^T\, \overline{P}_c\, y_c(x)$ as a polynomial Lyapunov function for \eqref{eqn20}, which is homogeneous in the entries of $x$ and of order $2c$.  
We implement the aforementioned procedure with Algorithm \ref{alg:one}, which specifically relies on CVX, a convex optimization toolbox built for use with MATLAB \cite{cvx, gb08}. Algorithm \ref{alg:one} takes as inputs the system parameters $A_1$ and $A_2$, and a positive integer $c$, and returns a matrix $\overline{P}_c$, in the case that one exists, which satisfies \eqref{eqn11} at the $c^{\text{th}}$ level.

\begin{algorithm}[t!]
\caption{Computing Meta-Lyapunov Functions}
\begin{algorithmic}[1]
\setlength\tabcolsep{0pt}
\Statex\begin{tabulary}{\linewidth}{@{}LLp{6cm}@{}}
\textbf{input}&:\:\:& $A_1,\, A_2 \in \mathbb{R}^{2\times 2}$ from \eqref{eqn1}. $c \in \N_{\geq 1}$.\\
\textbf{output}&:\:\:& $\overline{P}_{c} \in \mathbb{R}^{(c+1) \times (c+1)}$ satisfying \eqref{eqn17}.\\
&&
\end{tabulary}
\Function{MetaLyapunov}{$A_1,\ A_2,\, c$}
\State \textbf{Initialize: } Compute $\mathcal{A}_{c, 1}$ and $\mathcal{A}_{c, 2}$ by \eqref{eqn12}
\State \qquad\qquad\:\; Compute $W_c$ by \eqref{eqn19}
\State $B_{c, 1} \gets W_c^+ \mathcal{A}_{c, 1} W_c$
\State $B_{c, 2} \gets W_c^+ \mathcal{A}_{c, 2} W_c$
\State \textbf{cvx\_begin sdp}
\State \textbf{variable} $\overline{P}_{c}(c + 1, c+1)$ \textbf{semidefinite}
\State $\,0 \:\:> B_{c, 1}^T\, \overline{P}_{c} + \overline{P}_{c}\,B_{c, 1}$
\State $\,0 \:\:> B_{c, 2}^T\, \overline{P}_{c} + \overline{P}_{c}\,B_{c, 2}$
\State $\overline{P}_{c} > I_n$
\State \%\% Possibly Insert Objective Function 
\State \textbf{cvx\_end}
\If{Program feasible}
\State \textbf{return} $\overline{P}_{c}$
\Else
\State \textbf{return} `infeasible'
\EndIf
\EndFunction
\State\textbf{end function}
\end{algorithmic}
\label{alg:one}
\end{algorithm}

Note that Algorithm \ref{alg:one} computes the solution to a feasibility problem, rather than an optimization problem; that is, while Algorithm \ref{alg:one} searches for a $\overline{P}_c$ that satisfies the meta-Lyapunov constraint \eqref{eqn17}, this solution is computed without referencing any objective function.  Note however, that in the instance that multiple feasible solutions exist, it is preferable to choose $\overline{P}_c$ such that the sublevel sets of the resulting homogeneous Lyapunov function $V_{c}(x) = y_c(x)^T\, \overline{P}_c\, y_c(x)$ are small; this is due to the fact that $V_c(x)$ can be used to find infinite time reachable sets of \eqref{eqn20} under arbitrary switching.  For this reason, it is desirable to compute $\overline{P}_c$ as the solution to an optimisation problem, rather than a feasibility problem.

Little is known, in general about how one can relate the parameters of a polynomial to the volume of its sublevel sets. In our case as well, it is difficult to associate a metric of optimality with the a feasible solution to the meta-Lyapunov constraints \eqref{eqn17}.  Through experimentation, we have generally found that it is preferable to compute numerous solutions using different objective functions, and then compute an invariant region as the intersection of their respective sublevel sets.  Specifically we recommend using either using the objective function

\begin{algorithm}[h!]
\begin{algorithmic}[1]
\setlength\tabcolsep{0pt}
\setcounter{ALG@line}{10}
\State \textbf{minimize} $\overline{P}_c(1,\, 1)$
\end{algorithmic}
\end{algorithm}

\noindent{}which minimises the coefficient on $x_1^{2c}$ in the resultant Lyapunov function $V_{c}(x)$, or 

\begin{algorithm}[h!]
\begin{algorithmic}[1]
\setlength\tabcolsep{0pt}
\setcounter{ALG@line}{10}
\State \textbf{minimize} $\overline{P}_c(c+1,\,c+1)$
\end{algorithmic}
\end{algorithm}
\noindent{}which minimises the coefficient on $x_2^{2c}$.  These objective functions are provided in psuedocode, such that they can easily be inserted in Algorithm \ref{alg:one} at line 11.

\subsection{Numerical Results and Comparison with SOSTOOLS}
We now return to the example system \eqref{eqn20}, and compute feasible meta-Lyapunov functions with Algorithm \ref{alg:one}.
Additionally, we compute an over approximation of the infinite time reachable set of \eqref{eqn20} when beginning from the initial conditions $x(0) = [1,\,0]^T$.

As discussed in the preceding, we compute these invariant sets by implementing Algorithm \ref{alg:one}, while attempting to minimize $\overline{P}(1, 1)$, i.e. the coefficient on $x_1^{2c}$; see Algorithm \ref{alg:one}, Line 11.  This procedure was computed in MATLAB 2019b using CVX.

In the case of this example, Algorithm \ref{alg:one} was computed for $c \in \{1, 2, \cdots, 13\}$, thus generating homogeneous polynomial Lyapunov functions for all even orders between 2 and 26.  These Lyapunov functions were then used to calculate invariant regions of the state space using \eqref{idk}; see Figure \ref{fig:cs1}. 
Note that as the order of the meta-Lyapunov function increases, the derived invariant sets shrink in volume.   Further, certain higher-order the meta-Lyapunov functions were shown to have non-convex sublevel sets. We provide the number of solver iterations for each experiment, as well as the computations times, in Figure \ref{table:1}.

\begin{figure}[t]
    \centering
    \input{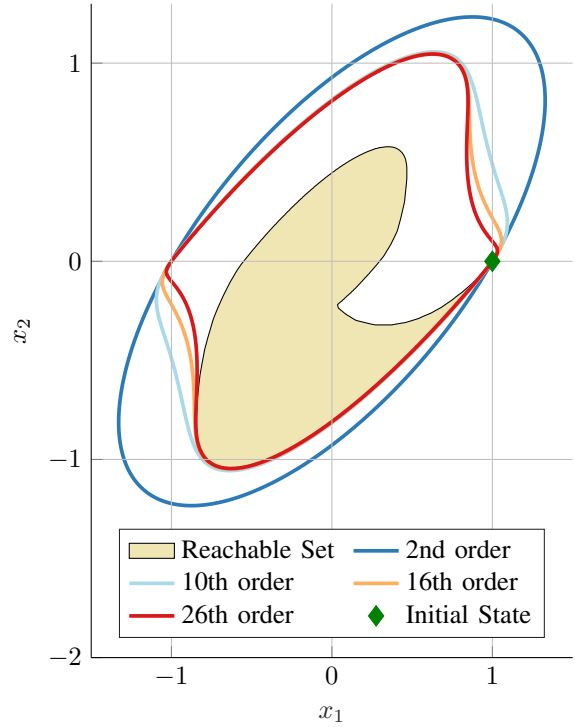}
    \caption{Simulated system response of \eqref{eqn20}.  When starting from $x_0 = [1, \: 0] ^T$, the system can only reach the region shown in light yellow, which was computed via simulation.  The equipotentials of high-order homogeneous Lyapunov functions are also shown.  Specifically, the dark blue, light blue, orange and red regions represent invariant sets calculated using quadratic, 10$^{\text{th}}$-order, 16$^{\text{th}}$-order and 26$^{\text{th}}$-order meta-Lyapunov functions, respectively.  The invariance of these regions is shown by \eqref{idk}.}
    \label{fig:cs1}
\end{figure}

We next compare the meta-Lyapunov function based method to the more traditional sum-of-squares based approach for calculating invariant regions.  Specifically, we search for high-order homogeneous polynomial Lyapunov functions for the system \eqref{eqn20} using SOSTOOLS, MATLAB's sum-of-squares toolbox \cite{sostools}.  We similarly implement SOSTOOLS with the solver SDPT3 and attempt to generate homogeneous polynomial Lyapunov functions for the system \eqref{eqn20} while minimizing the coefficient on $x_1^{2c}$.
As was the case previously, we provide the number of solver iterations for each experiment, as well as the computations times (See Figure \ref{table:1}).

In the experiment, SOSTOOLS was only able to generate homogeneous polynomial Lyapunov functions of order 10 or below; a solver error was returned during each search for more complex Lyapunov function. 
In contrast, Algorithm \ref{alg:one}, when implemented through CVX, was able to generate up to $26^{\text{th}}$-order polynomial Lyapunov functions while minimizing the same objective function.  We attribute this discrepancy to the fact that many of the steps required in a traditional sum-of-squares based search optimization are not required by Algorithm \ref{alg:one}.  For example, Algorithm \ref{alg:one} does not compute the time rate of change of the monomials in $x$ of order $2c$; that is, Algorithm \ref{alg:one} begins with a closed from representation of $\dot{y}_c(x)$, which is encoded in the matrices $B_{c, 1}$ and $B_{c, 2}$.  SOSTOOLS must compute $\dot{y}_c(x)$ online as a sum of squares of lower-order monomials and,  for this reason, one can expect SOSTOOLS to preform less efficiently.  

Despite this, in the experiments where SOSTOOLS was able to correctly generate homogeneous Lyapunov functions, we found that SOSTOOLS preformed faster in computation that the meta-Lyapunov search implemented with CVX; however, it did take SDPT3 more solver iterations to generate the solution which minimized the objective function when implemented through SOSTOOLS (Figure \ref{table:1}).

\begin{figure}
	\centering
	\includegraphics[width = 0.42\textwidth]{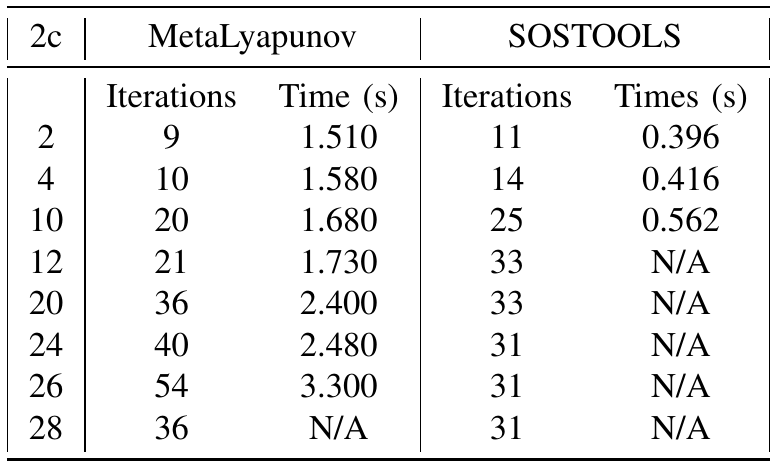}
	\caption{Comparing Algorithm \ref{alg:one} to a search for homogeneous Lyapunov functions using SOSTOOLS. Algorithm \ref{alg:one} is used to compute homogeneous Lyapunov functions of orders 2, 4, 10, 12, 20, 24, and 26 for the system \eqref{eqn20}.  SOSTOOLS, however, is only able to correctly generate Lyapunov functions of order 10 and below. The computation time and number of solver iterations are provided for each experiment. The symbol N/A is used when the solver is unable to find a homogeneous Lyapunov function of a certain order; in this case, the number of solver iterations which were preformed before failure is also provided.}
	\label{table:1}
\end{figure}

\section{Conclusion}
This work addresses the problem of searching for homogeneous polynomial Lyapunov functions for stable switched linear systems.
An equivalence is shown between polynomial Lyapunov functions for switched linear systems and quadratic Lyapunov functions for a related hierarchy of Lyapunov differential equations.  A computationally competitive algorithm is presented for generating high-order homogeneous polynomial Lyapunov functions in this manner.

\bibliography{Bibliography}

\begin{thebibliography}{10}

\bibitem{liberzon2003switching}
D.~Liberzon, {\em Switching in Systems and Control}.
\newblock Systems \& Control: Foundations \& Applications, Birkh{\"a}user
  Boston, 2003.

\bibitem{feron1996quadratic}
E.~Feron, {\em Quadratic stabilizability of switched systems via state and
  output feedback}.
\newblock Center for Intelligent Control Systems, 1996.

\bibitem{4806347}
R.~{Goebel}, R.~G. {Sanfelice}, and A.~R. {Teel}, ``Hybrid dynamical systems,''
  {\em IEEE Control Systems Magazine}, vol.~29, pp.~28--93, April 2009.

\bibitem{Brockett1993}
R.~W. Brockett, {\em Hybrid Models for Motion Control Systems}, pp.~29--53.
\newblock Boston, MA: Birkh{\"a}user Boston, 1993.

\bibitem{yak1}
V.~A. Yakubovich, ``The solution of some matrix inequalities encountered in
  automatic control theory,'' {\em Societ Math Dokl.}, vol.~143, no.~5,
  pp.~652--656, 1964.

\bibitem{BEFB:94}
S.~Boyd, L.~{El~{G}haoui}, E.~Feron, and V.~Balakrishnan, {\em Linear Matrix
  Inequalities in System and Control Theory}, vol.~15 of {\em Studies in
  Applied Mathematics}.
\newblock Philadelphia, PA: {SIAM}, June 1994.

\bibitem{Barmish1985}
B.~R. Barmish, ``Necessary and sufficient conditions for quadratic
  stabilizability of an uncertain system,'' {\em Journal of Optimization Theory
  and Applications}, vol.~46, pp.~399--408, Aug 1985.

\bibitem{yoon2019}
Y.~Yoon, C.~Klett, and E.~Feron, ``Bounding the state covariance matrix for a
  randomly switching linear system with noise,'' {\em arXiv preprint
  arXiv:1905.09427}, 2019.

\bibitem{yak2}
V.~Yakubovich, ``The method of matrix inequalities in the stability theory of
  nonlinear control systems,'' {\em Automation and Remote Control}, vol.~26,
  pp.~577--592, 1965.

\bibitem{yak3}
V.~A. Yakubovich, ``Frequency conditions for the existence of absolutely stable
  periodic and almost periodic limiting regimes of control systems with many
  nonstationary elements,'' {\em IFAC World Congress}, 1966.

\bibitem{rantzer1997}
M.~Johansson and A.~Rantzer, ``Computation of piecewise quadratic lyapunov
  functions for hybrid systems,'' in {\em 1997 European Control Conference
  (ECC)}, pp.~2005--2010, IEEE, 1997.

\bibitem{parrilo2000}
P.~A. Parrilo, {\em Structured semidefinite programs and semialgebraic geometry
  methods in robustness and optimization}.
\newblock PhD thesis, California Institute of Technology, 2000.

\bibitem{parrilo2003semidefinite}
P.~A. Parrilo, ``Semidefinite programming relaxations for semialgebraic
  problems,'' {\em Mathematical programming}, vol.~96, no.~2, pp.~293--320,
  2003.

\bibitem{CPLF}
P.~Mason, U.~Boscain, and Y.~Chitour, ``Common polynomial lyapunov functions
  for linear switched systems,'' {\em SIAM journal on control and
  optimization}, vol.~45, no.~1, pp.~226--245, 2006.

\bibitem{6161493}
A.~A. {Ahmadi} and P.~A. {Parrilo}, ``Converse results on existence of sum of
  squares {Lyapunov} functions,'' in {\em 2011 50th IEEE Conference on Decision
  and Control and European Control Conference}, pp.~6516--6521, Dec 2011.

\bibitem{SeDuMi}
J.~F. Sturm, ``Using sedumi 1.02, a matlab toolbox for optimization over
  symmetric cones,'' {\em Optimization Methods and Software}, vol.~11, no.~1-4,
  pp.~625--653, 1999.

\bibitem{cvx}
M.~Grant and S.~Boyd, ``{CVX}: Matlab software for disciplined convex
  programming, version 2.1,'' Mar. 2014.

\bibitem{sostools}
G.~V. S. P. P.~S. A.~Papachristodoulou, J.~Anderson and P.~A. Parrilo, {\em
  {SOSTOOLS}: Sum of squares optimization toolbox for {MATLAB}}.

\bibitem{gb08}
M.~Grant and S.~Boyd, ``Graph implementations for nonsmooth convex programs,''
  in {\em Recent Advances in Learning and Control}, Lecture Notes in Control
  and Information Sciences, pp.~95--110, Springer-Verlag Limited, 2008.

\end{thebibliography}
\bibliographystyle{ieeetr}
\end{document}